\documentclass[11pt]{article}
\usepackage[a4paper,margin=25mm]{geometry}
\usepackage{amsmath,amssymb,amsthm,natbib}
\usepackage[hidelinks]{hyperref}
\newtheorem{theorem}{Theorem}
\newtheorem{proposition}{Proposition}
\newcommand{\E}{\mathbb E}
\newcommand{\ip}[2]{\langle #1,#2\rangle}
\newcommand{\Risk}{\operatorname{Risk}}
\title{Unidentified prediction-risk rankings under spatial sampling with fixed inclusion probabilities}
\author{Shinichiro Shirota\\Hitotsubashi University\\\texttt{shinichiro.shirota@r.hit-u.ac.jp}}
\date{}
\begin{document}
\maketitle
\begin{abstract}
For Gaussian spatial prediction, we construct positive sampling mechanisms with identical observation laws and conditional site inclusion probabilities, yet opposite risk rankings for two fixed predictors. The construction selects two of four sites and allows independent Gaussian measurement error.
\end{abstract}

\section{Introduction}
Spatial prediction can depend on how sampling locations are selected \citep{diggle2010}. Whole-design models allow dependence between locations beyond that represented by a sampling intensity \citep{gray2023}. We ask whether correct conditional site inclusion probabilities, together with the complete distribution of the observed data, determine which of two specified predictors has smaller integrated squared-error risk.

We construct an example in which they do not. Two sites are selected from four candidates in a Gaussian field with a fixed, known covariance. Every site has conditional inclusion probability $1/2$, every pair has unconditional probability $1/6$, and the joint law of the selected pair and its measurements is constant along a family of positive mechanisms. Nevertheless, the risk ranking of a baseline posterior mean and a working posterior mean reverses. The construction remains valid with independent Gaussian measurement error of any fixed finite variance. Consequently, repeated independent observations from this experiment cannot identify that ranking, even when the field law is known.

Nonidentification under nonignorable missingness is established in pattern-mixture theory \citep{little1993}; coarsening theory also shows that testability depends on the observation structure \citep{cator2004}. Our contribution is a constructive spatial example satisfying these constraints simultaneously, rather than a new general nonidentification principle. A product of compactly supported derivatives preserves every pairwise observed distribution while changing a conditional field mean. Matching contrasts describe the remaining design freedom. A simpler, distinguishable family then separates reversal of a correction's direction from excessive correction magnitude.

\section{Observation structure and risk}
Let $P$ be a common baseline law for $(T,X,E)$, where $T$ is a square-integrable random element of a real separable Hilbert space and $E$ collects any observation noise. Potential observations are measurable maps $Y_r=f_r(X,E)$, indexed by a finite pattern set $\mathcal R$. After generating these variables, draw $R$ with probabilities $q_r(X)$, independently of $(T,E)$ conditional on $X$. Only $O=(R,Y_R)$ is observed. This latent-state factorization is not the missing-at-random condition of \citet{rubin1976}: $X$ need not be observed, and $q_r$ may depend on its unobserved coordinates.

Fix a working mechanism with $0<m\leq |\mathcal R|q^{\mathrm w}_r(X)\leq M<\infty$ almost surely for every $r$, and define the Bochner conditional means
\[
 I_r=\E[T\mid Y_r],\qquad
 J_r=\frac{\E[Tq^{\mathrm w}_r(X)\mid Y_r]}{\E[q^{\mathrm w}_r(X)\mid Y_r]},\qquad d_r=J_r-I_r.
\]
The denominator is bounded below by $m/|\mathcal R|$. Conditional Jensen gives $\E\|I_r\|^2\leq\E\|T\|^2$ and $\E\|J_r\|^2\leq(M/m)\E\|T\|^2$, so all branch risks below are finite. Write
\[
 \Risk_q(a)=\E_q\|T-a_R(Y_R)\|^2,\quad
 A_q=\E_q\|d_R\|^2,\quad C_q=\E_q\ip{T-I_R}{d_R}.
\]
For the same fixed pair $I,J$ under each mechanism,
\begin{equation}\label{eq:loss}
 \Delta_q:=\Risk_q(I)-\Risk_q(J)=2C_q-A_q,\qquad
 \Risk_q(I+wd)-\Risk_q(I)=w^2A_q-2wC_q.
\end{equation}
Negative $\Delta_q$ means that the working adjustment increases risk.

\section{Identical observations and opposite rankings}
Let $Z$ be a centered Gaussian field on $[0,1]$ with continuous covariance $k$. Take a jointly measurable version, available by mean-square continuity, and set $T=Z(\cdot)\in L^2[0,1]$ with Lebesgue measure. Fubini gives $\E\|T\|^2=\int_0^1 k(v,v)\,dv<\infty$. At distinct sites $S=(a_1,a_2,a_3,a_4)$ let $X=Z(S)$ have positive definite covariance $K$. Set $Y_i=X_i+\epsilon_i$, with independent $N(0,\tau^2)$ errors independent of the field and selection randomization, for fixed $0\leq\tau<\infty$. Use the pair order
\[
 \mathcal R=(12,34,13,24,14,23),\qquad s=(1,1,-1,-1,0,0).
\]
Thus $s_{12}=s_{34}=1$, $s_{13}=s_{24}=-1$, and $s_{14}=s_{23}=0$.

\begin{theorem}\label{thm:equiv}
For every such field, every fixed finite $\tau\geq0$, and every $0<\kappa<1$, there is $h\in C_c^\infty(\mathbb R^4)$ with $|h|\leq1$ such that the mechanisms
\begin{equation}\label{eq:mechanism}
 q^t_r(X)=\{1+t\kappa s_rh(X)\}/6,\qquad -1\leq t\leq1,
\end{equation}
satisfy $q^t_r\geq(1-\kappa)/6$, $\Pr_t(i\in R\mid X)=1/2$, and
\begin{equation}\label{eq:observed}
 P_t(R=r,Y_r\in dy)=P(Y_r\in dy)/6.
\end{equation}
Let $I_r=\E[T\mid Y_r]$ and $J_r=\E_1[T\mid R=r,Y_r]$. For these fixed rules there is $0<A_\tau<\infty$ such that
\begin{equation}\label{eq:ranking}
 \E_t[T\mid R=r,Y_r]=I_r+t d_r,\qquad
 A_{q^t}=A_\tau,\qquad \Delta_{q^t}=(2t-1)A_\tau.
\end{equation}
In particular, $\Delta_{q^{-1}}=-3A_\tau$ and $\Delta_{q^1}=A_\tau$, although the observed-data laws coincide.
\end{theorem}
\begin{proof}
Choose nonzero nonnegative $\psi_i\in C_c^\infty(\mathbb R)$ and write $B_i=\int\psi_i>0$. With $\phi_K$ the Gaussian density, put
\[
 F(x)=\psi'_1(x_1)\psi'_2(x_2)\psi'_3(x_3)\psi_4(x_4),\qquad
 h(x)=cF(x)/\phi_K(x).
\]
The density has a positive minimum on the compact support of $F$, so a sufficiently small $c>0$ makes $|h|\leq1$. This quotient is globally smooth and compactly supported. The asymmetry is essential: every pair leaves at least one of coordinates $1,2,3$ unobserved. Since $\int\psi'_i=0$, integration over that coordinate yields $\E[h\mid X_r]=0$ for every pair, without requiring independent coordinates of $X$. Independence of the measurement errors then gives $\E[h\mid Y_r]=0$.

The incident sums of $s$ vanish and $\sum_r s_r=0$. This proves normalization and the inclusion constraints. Conditional expectation of \eqref{eq:mechanism} given $Y_r$ proves \eqref{eq:observed}, and its posterior denominator equals one after multiplication by six. Thus
\[
 d_r=\kappa s_r\E[Th\mid Y_r],\qquad
 \E_t[T\mid R=r,Y_r]=I_r+td_r.
\]
Under the common observation law, $A_{q^t}=A_\tau$ and $C_{q^t}=tA_\tau$; \eqref{eq:loss} proves the risk formula. Boundedness of $h$ and conditional Jensen give finiteness.

For strict positivity, let $f_{12}$ denote the density of $(Y_1,Y_2)$ and define $b_i=\psi'_i$ if $\tau=0$, or $b_i=\psi'_i*\phi_\tau$ if $\tau>0$, where $\phi_\tau$ is the normal error density. Integration by parts gives
\begin{equation}\label{eq:moment}
 f_{12}(y)\E[Xh\mid Y_{12}=y]
 =-cB_3B_4b_1(y_1)b_2(y_2)e_3.
\end{equation}
Here $e_3$ is the third coordinate vector; the other components contain a vanishing integral of a derivative. Each $b_i$ is continuous and nonzero: for positive $\tau$, Gaussian convolution is injective because its Fourier multiplier never vanishes. Their product is nonzero on an open set of positive observation probability. Gaussian regression gives $\E[Th\mid Y_{12}](v)=H(v)\E[Xh\mid Y_{12}]$, where $H(v)=k(v,S)K^{-1}$ is a row vector. Its third component is continuous and $H_3(a_3)=1$, so $\int_0^1 H_3(v)^2\,dv>0$, using a one-sided neighborhood at an endpoint if needed. Pair 12 therefore contributes strictly positively to $A_\tau$.
\end{proof}

The parameter $\kappa$ specifies the working tilt strength, while $t$ measures the true tilt relative to that working model ($t=1$). Although the sampling probabilities depend on $t\kappa$, varying $t$ holds the working predictor fixed. Equation~\eqref{eq:observed} implies more than failure of inclusion summaries: any binary ranking decision based on any finite number of independent copies of $O$, with independent randomization allowed, has average error $1/2$ at the two endpoints with equal prior weights. Indeed, its probability of choosing $J$ is the same at both endpoints, where opposite choices are correct. Re-estimation from these observations cannot resolve this ambiguity. This assertion concerns the specified experiment and pair of predictors, not arbitrary prediction problems or repeated sampling of one shared field.

\section{Matching contrasts and a simpler path}
The design freedom behind the example has a short algebraic description. For the general setup of Section 2, define
\[
 g_r(x)=\E[2\ip{T-I_r}{d_r}-\|d_r\|^2\mid X=x],\qquad
 A^{\mathrm w}=A_{q^{\mathrm w}}.
\]
Bayes orthogonality under $q^{\mathrm w}$ and \eqref{eq:loss} imply
$\Delta_q=A^{\mathrm w}+\E\sum_r(q_r-q^{\mathrm w}_r)g_r$.
Let $M_1=(12,34)$, $M_2=(13,24)$, $M_3=(14,23)$ and $G_j=\sum_{r\in M_j}g_r$.
\begin{proposition}\label{prop:matching}
Two size-two mechanisms on four candidates have identical conditional site inclusions if and only if their difference, in the stated pair order, is
$u=(v_1,v_1,v_2,v_2,v_3,v_3)$ with $\sum_jv_j=0$ almost surely. Given a normalized working mechanism, feasibility is equivalent to
$v_j(x)\geq-\min_{r\in M_j}q^{\mathrm w}_r(x)$ for every $j$, for $P_X$-almost every $x$, together with this sum constraint. Then
\begin{equation}\label{eq:matching}
 \Delta_q=A^{\mathrm w}+\E\{v_1(G_1-G_3)+v_2(G_2-G_3)\}.
\end{equation}
Unconditional pair frequencies also agree exactly when $\E v_j=0$ for all $j$.
\end{proposition}
\begin{proof}
The four incident-edge sums of $u$ vanish. Adding those for vertices 1 and 2 and subtracting those for 3 and 4 gives $u_{12}=u_{34}$; the other opposite-edge equalities follow similarly. One vertex equation gives $\sum v_j=0$. Conversely these relations preserve every inclusion sum and normalization. The bounds are exactly nonnegativity of the six probabilities. Substitution gives \eqref{eq:matching}.
\end{proof}
\noindent\textit{Remark.} If $|v_j|\leq\rho_j$ for $j=1,2$, with fixed nonnegative $\rho_j$, then $\sum_{j=1}^2\rho_j\E|G_j-G_3|\leq A^{\mathrm w}$ suffices for $\Delta_q\geq0$. This immediate bound is a latent-mechanism sensitivity calculation, not a risk certificate estimable from the observed data.

For comparison, replace $h(X)$ by $\tanh X_3$ in \eqref{eq:mechanism} and observe $Y_r=X_r$ without error. Redefine $I,J$ for this family, taking $q^1$ as working, and then hold them fixed. As shown in the supplement, with $U=A_{q^0}>0$ and $L=A_{q^1}>0$,
\begin{equation}\label{eq:tanh}
 A_{q^t}=U+t(L-U),\qquad C_{q^t}=tL,\qquad
 \Delta_{q^t}=t(U+L)-U.
\end{equation}
Here $U$ is the mean squared adjustment under uniform selection and $L$ its working-law counterpart. Full correction is harmful for $t<t_*=U/(U+L)$. Every positive correction weight is harmful when $t\leq0$, whereas for $0<t<t_*$ small positive weights improve risk even though full correction is harmful. Unlike Theorem~\ref{thm:equiv}, this family generally changes the observation law. For example, pair 34 has joint density proportional to $1+t\kappa\tanh x_3$. Its ambiguity is therefore confined to the unchanged inclusion summaries.

\section{Scope and implications}
Theorem~\ref{thm:equiv} isolates dependence between the selected configuration and the latent field while each individual inclusion indicator remains independent of the field. The finite candidate design is not an intensity-driven log-Gaussian Cox process. It illustrates why specifying site inclusions alone does not specify a whole-design model such as that of \citet{gray2023}; it neither proves that every such model is unidentified nor supplies a practical monitoring protocol using unobserved field values.

Intensity-based adjustment has been studied by \citet{hsiao2026}, including prediction under design misspecification. Our example does not imply that their fitted corrections vanish, or that a particular preferential-sampling test has no power throughout a larger model class. The identical-law argument applies specifically to the family in Theorem~\ref{thm:equiv}. Likewise, the unbiasedness and consistency results of \citet{lu2025} concern regression and covariance estimation under specified spatial point-process assumptions. They address different estimands and assumptions from the integrated risk ranking here.

The construction establishes a strict existence result, not a practically large risk gap. No lower bound uniform in measurement-error variance is asserted. Additional restrictions on the selection mechanism or external validation observations may identify the ranking. The present result identifies the limitation of the stated information, without prescribing which adjustment should be used in an application.

\paragraph{Data and code availability.}
No external empirical data were used. The accompanying supplementary code \texttt{checks.py} and \texttt{results.json} provide exact finite verification examples. These discrete examples do not numerically evaluate the Gaussian construction or establish a practical effect size.
\paragraph{Competing interests.}
The author declares no competing interests relevant to this work.
\paragraph{Funding.}
This work was supported by JSPS KAKENHI Grant Numbers JP24K00244 and JP26K14728.
\paragraph{Supplementary material.}
The supplement gives conditional-moment calculations, the proof of the auxiliary risk path, and reproduction instructions.
\paragraph{AI-assisted preparation.}
During the preparation of this work, the author used generative AI tools, including ChatGPT, Claude, and Codex, to assist with research discussions, mathematical analysis, literature review, computational verification, and manuscript preparation. The author reviewed and edited the AI-assisted material and takes full responsibility for the content of this work.
\begingroup
\raggedright
\bibliographystyle{plainnat}
\bibliography{references}
\endgroup
\clearpage
\appendix
\setcounter{section}{0}\setcounter{equation}{0}\setcounter{proposition}{0}
\renewcommand{\thesection}{S\arabic{section}}
\renewcommand{\theequation}{S\arabic{equation}}
\renewcommand{\theproposition}{S\arabic{proposition}}
\section*{Supplement}

\section{Risk identity and two-candidate contrast}
All conditional expectations in the kernels use the common baseline distribution, before selection. Write $k_r(x)=\E[\langle T-I_r,d_r\rangle\mid X=x]$, $a_r(x)=\E[\|d_r\|^2\mid X=x]$, and $g_r=2k_r-a_r$. These loss kernels are distinct from the field covariance $k(v,w)$ below. For a branch integrand $F_r$,
\[
 \E_q F_R=\E\sum_r q_r(X)\E[F_r\mid X].
\]
Expanding $\|T-I-d\|^2$ gives $\Delta_q=2C_q-A_q$. Under the working law, $\E_{q^{\mathrm w}}[T-J\mid R,Y_R]=0$, hence $C_{q^{\mathrm w}}=A_{q^{\mathrm w}}=\Delta_{q^{\mathrm w}}$. Subtracting that identity from the general one gives the risk contrast used in Section 4 of the main text. This proof uses no Gaussian assumption.

For patterns $(00,10,01,11)$, preservation of both conditional first-order inclusion probabilities and normalization forces
\[
 u=\delta(X)(1,-1,-1,1).
\]
Indeed, $u_{10}+u_{11}=u_{01}+u_{11}=0$, and summing the four components gives $u_{00}=u_{11}$. Feasibility is exactly
\[
 -\min(q^{\mathrm w}_{00},q^{\mathrm w}_{11})
 \leq\delta\leq
 \min(q^{\mathrm w}_{10},q^{\mathrm w}_{01}).
\]
For $G=g_{00}-g_{10}-g_{01}+g_{11}$ and
$H=k_{00}-k_{10}-k_{01}+k_{11}$,
\[
 \Delta_q=A^{\mathrm w}+\E[\delta G],\qquad
 C_q=A^{\mathrm w}+\E[\delta H].
\]
Unconditional pattern frequencies agree exactly when $\E\delta=0$.
Then $\E\delta G=\operatorname{Cov}(\delta,G)$. An additive branch function has zero alternating contrast, so it is invisible to these perturbations. These identities characterize this two-candidate design separately from the fixed-size four-candidate design of the main text.

\section{Gaussian regression and positive correction}
Use the main text's sites $S=(a_1,a_2,a_3,a_4)$ and pair order. Write
$H(v)=k(v,S)K^{-1}$, where $K=\operatorname{Cov}(X)$ and $H(v)$ is a $1\times4$ row vector.
Gaussian regression gives $Z(v)=H(v)X+V(v)$, with centered residual $V$ independent of $X$. The covariance with a bounded tilt can therefore be computed through the candidate vector.

For the auxiliary $\tanh X_3$ family, each incident sum of $s$ is zero and $\E\tanh X_3=0$ by centered Gaussian symmetry. Hence all conditional site inclusions equal $1/2$ and all unconditional pair frequencies equal $1/6$. The likelihood ratios are in $[1-\kappa,1+\kappa]$.
For the pair $12$, put $W=X_3$ and condition on $X_1,X_2$. The working adjustment is
\[
 d_{12}(v)=\frac{\kappa\operatorname{Cov}(Z(v),\tanh W\mid X_1,X_2)}
 {1+\kappa\E[\tanh W\mid X_1,X_2]}.
\]
Its denominator is at least $1-\kappa>0$.
The residual variance $\sigma^2>0$ follows from positive definiteness of $K$. For a normal variable with mean $\mu$ and variance $\sigma^2$, integration by parts gives
\[
 \E[(W-\mu)\tanh W]=\sigma^2\E[\operatorname{sech}^2 W]>0.
\]
Boundary terms vanish because $\tanh$ is bounded and the normal density decays.
Conditional Gaussian regression also gives
\[
 \operatorname{Cov}(Z(v),\tanh W\mid X_1,X_2)
 =\operatorname{Cov}(Z(v),W\mid X_1,X_2)
   \E[\operatorname{sech}^2W\mid X_1,X_2].
\]
Explicitly, with $S_{12}=(a_1,a_2)$,
\[
 \operatorname{Cov}(Z(v),W\mid X_1,X_2)
 =k(v,a_3)-k(v,S_{12})K_{12,12}^{-1}K_{12,3}.
\]
This is continuous in $v$ and equals $\sigma^2$ at $a_3$.
The factor $\E[\operatorname{sech}^2W\mid X_1,X_2]$ is positive. A positive-measure neighborhood, one-sided at an endpoint if necessary, therefore has nonzero correction. Evaluation at a single point alone would not establish an $L^2$ assertion; continuity and Lebesgue measure are essential here.

All branch likelihood ratios relative to uniform selection lie in
$[1-\kappa,1+\kappa]$. Baseline second moments are finite by continuous covariance on a compact interval. Conditional Jensen and these bounds ensure finite $I,J,U,L$. In particular $L\geq(1-\kappa)U>0$ once $U>0$ is established.

For a fixed rule pair, expectations of every integrable branch loss are affine in $t$.
At $t=0$, $I$ is Bayes and $C_0=0$; at $t=1$, $J$ is Bayes and $C_1=L$.
This proves $C_t=tL$ without differentiating a posterior or assuming a fixed posterior normalizing constant.
The full correction threshold lies strictly between zero and one.
Here $A_t=U+t(L-U)\geq(1-\kappa)U>0$ and $t_*=U/(U+L)$. At $t=0$ every nonzero weight is harmful; at $t=t_*$ the two unweighted rules $I,J$ tie.
For $0<t<t_*$ improvement occurs for $0<w<2tL/A_t$.
For $t\leq0$, $w^2A_t-2wtL>0$ for every $w>0$.

\section{Details of the observation-equivalent construction}
This section expands the moment calculations in Theorem 1 of the main text.
The following calculation verifies the main text's construction. It retains the same field law and adds measurement errors $Y_i=X_i+\epsilon_i$, where the errors are mutually independent $N(0,\tau^2)$, independent of the field and selection, for a fixed finite $\tau\geq0$. The rules are redefined for this experiment and then frozen.

To ensure that every selected pair leaves an unobserved derivative coordinate, derivatives are used at coordinates 1, 2, 3 but not 4. Choose nonzero, nonnegative $\psi_i\in C_c^\infty(\mathbb R)$.
Let $\phi_K$ be the positive Gaussian density of $X$, and set
\[
 F(x)=\psi_1'(x_1)\psi_2'(x_2)\psi_3'(x_3)\psi_4(x_4),
 \qquad h(x)=cF(x)/\phi_K(x).
\]
On the compact support of $F$, $\phi_K$ has positive minimum and $F$ has finite maximum absolute value. Choose $c>0$ small enough that $|h|\leq1$. The quotient is smooth everywhere; there is no discontinuous extension at the support boundary. Write $B_i=\int\psi_i>0$.

For every selected pair $r$, at least one of the first three coordinates is unobserved. Since $\int\psi_i'=0$, cancellation of $\phi_K$ and Fubini yield
\[
 \phi_{K,r}(x_r)\E[h\mid X_r=x_r]
   =c\int F(x)\,dx_{-r}=0.
\]
This does not assume that the Gaussian coordinates are independent.
Because the errors are independent of $X$, $\E[h\mid X_r,Y_r]=\E[h\mid X_r]=0$; the tower property therefore gives $\E[h\mid Y_r]=0$. Conditioning on $Y_r$ alone is not a nested conditioning argument without this independence step.

Use the main text's pair order $(12,34,13,24,14,23)$ and $s=(1,1,-1,-1,0,0)$. For $-1\leq t\leq1$ and $0<\kappa<1$, put $q^t_r=(1+t\kappa s_rh)/6$.
Normalization, positivity, size two, and conditional inclusion $1/2$ hold.
Moreover
\[
 P_t(R=r,Y_r\in dy)=P(Y_r\in dy)/6
\]
for every $t$. Define $I_r=\E[T\mid Y_r]$ and
$d_r=\kappa s_r\E[Th\mid Y_r]$. The posterior mean under $q^t$ is exactly $I_r+td_r$, and $J=I+d$ is the working posterior at $t=1$. Thus, under the common observation law,
\begin{equation}
 A=\E\|d_R\|^2,\qquad C_t=tA,\qquad\Delta_t=(2t-1)A. \label{eq:equiv}
\end{equation}
Bounded $h$ and conditional Jensen prove finiteness.

To prove strict positivity, first take $\tau=0$. Direct integration by parts at pair 12 gives
\[
 \phi_{K,12}(x_1,x_2)\E[Xh\mid X_1=x_1,X_2=x_2]
 =-cB_3B_4\psi_1'(x_1)\psi_2'(x_2)e_3.
\]
Components 1 and 2 contain $\int\psi_3'=0$; component 4 also contains that factor. Component 3 uses $\int x\psi_3'(x)\,dx=-B_3$.
The two nonzero derivatives are continuous, so the right side is nonzero on a nonempty open set.

For $\tau>0$, let $\phi_\tau$ denote the $N(0,\tau^2)$ density and put $b_i=\psi_i'*\phi_\tau$. With $f_{12}$ the everywhere positive density of $(Y_1,Y_2)$, direct integration gives
\[
 f_{12}(y)\E[Xh\mid Y_{12}=y]
 =-cB_3B_4b_1(y_1)b_2(y_2)e_3.
\]
Gaussian convolution is injective: its Fourier multiplier
$\exp(-\tau^2\omega^2/2)$ never vanishes. Each $b_i$ is thus nonzero and continuous, and their product is nonzero on an open set.
Regression transfers the moment to
$H_3(v)=[k(v,S)K^{-1}]_3$. Since $H_3(a_3)=1$ and $H_3$ is continuous, its squared integral is positive.
The residual $V$ is independent of $X$ and the measurement error and centered, so its contribution vanishes. Hence $A>0$ for every fixed finite $\tau$.
No lower bound uniform in $\tau$ is asserted.

At $t=-1$ and $t=1$, respectively, the risk differences in \eqref{eq:equiv} are $-3A$ and $A$, although the observation laws coincide. These are differences of nonnegative risks, not negative risks. Any binary choice based only on any number of independent observation replicates has average error $1/2$ under equal endpoint prior probabilities. The same argument allows an auxiliary randomization independent of the data. It is a standard two-point consequence. A procedure allowed to abstain can avoid false declarations by abstaining; it cannot thereby produce uniformly correct rankings.

\section{Verification and reproduction}
The accompanying \texttt{checks.py} enumerates a finite sign-valued analogue using rational arithmetic. It separately sums actual squared losses and verifies the affine path, inclusion constraints, matching contrasts, posterior means under observation equivalence, and two-endpoint ranking errors. All six patterns are included. This deterministic check uses no random seed and is an implementation check, not a proof of the Gaussian results or an estimate of their effect size.

Run \texttt{python3 checks.py} and compare the resulting \texttt{results.json} with the supplied output, disregarding the Python version field if different. The only dependencies are Python's standard library. The source package includes the code and output in its \texttt{anc/} directory. Compile the main manuscript with BibTeX for its references and compile this supplement separately, or build the integrated \texttt{main.tex} in the arXiv source package. The accompanying README gives the tested build command.
\end{document}